\documentclass[letterpaper]{article} %
\usepackage{aaai23}  %
\usepackage{times}  %
\usepackage{helvet}  %
\usepackage{courier}  %
\usepackage[hyphens]{url}  %
\usepackage{graphicx} %
\usepackage{natbib}  %
\usepackage{caption} %
\usepackage{algorithm}
\usepackage{algorithmic}

\usepackage{newfloat}
\usepackage{listings}
\DeclareCaptionStyle{ruled}{labelfont=normalfont,labelsep=colon,strut=off} %
\floatstyle{ruled}
\newfloat{listing}{tb}{lst}{}
\floatname{listing}{Listing}
\usepackage[utf8]{inputenc}
\usepackage{dblfloatfix} 
\usepackage{booktabs}       %
\usepackage{amsfonts}       %
\usepackage{nicefrac}       %
\usepackage{microtype}      %
\usepackage{amsmath,amsthm}
\usepackage{cleveref}
\usepackage{subcaption}
\usepackage{enumitem}
\usepackage{color}
\usepackage{comment}
\usepackage{mathtools}
\usepackage [english]{babel}
\usepackage [autostyle, english = american]{csquotes}
\MakeOuterQuote{"}
\usepackage{graphics} %
\usepackage{epsfig} %
\usepackage{amssymb}  %

\newtheorem{theorem}{Theorem}
\newtheorem{lemma}{Lemma}

\usepackage{color} %

\title{Fine-Grained Complexity Analysis of Multi-Agent Path Finding on 2D Grids}
\author{
    Tzvika Geft%
}
\affiliations{
    The Blavatnik School of Computer Science, Tel Aviv University, Israel

    zvigreg@mail.tau.ac.il
}

\begin{document}

\newcommand{\poslit}{\ensuremath{R^+}}
\newcommand{\neglit}{\ensuremath{R^-}}
\newcommand{\true}{\textsc{true}}
\newcommand{\false}{\textsc{false}}
\newcommand{\litrob}{\ensuremath{r^j}}
\newcommand{\src}[1] {\ensuremath{s(#1)}}
\newcommand{\trg}[1] {\ensuremath{t(#1)}}
\newcommand{\pth}[1] {\ensuremath{\pi_{#1}}}
\newcommand{\pthdef}[2] {\ensuremath{\pi_{#1}(#2)}}
\newcommand{\gadg}{\ensuremath{\Xi}}
\newcommand{\spac}{\ensuremath{\delta}-spacious}

\newcommand{\din}{\delta_\textrm{in}}
\newcommand{\dout}{\delta_\textrm{out}}

\newcommand{\bef}{B}
\newcommand{\aft}{A}

\let\ov\overline
\newcommand{\clause}[3]{\ensuremath{(#1 \lor #2 \lor #3)}}
\newcommand{\lits}{\ensuremath{\ell_i, \ell_j, \ell_k}}
\newcommand{\cl}{\clause{\ell_i}{\ell_j}{\ell_k}}

\newcommand{\bestcost}{\mathop{d^*}}

\newcommand{\set}[1]{\ensuremath{\{#1\}}}

\definecolor{gray}{rgb}{0.35,0.35,0.35}
\definecolor{blue}{rgb}{0,0,1}
\definecolor{red}{rgb}{1,0,0}
\definecolor{orange}{rgb}{0.75, 0.4, 0}
\definecolor{green}{rgb}{0.0, 0.5, 0.0}
\newcommand{\aviv}[1]{{\color{green}\textbf{Aviv: }\sf#1}}
\newcommand{\tzvika}[1]{{\color{blue}\textbf{Tzvika: }\sf#1}}
\newcommand{\ken}[1]{{\color{orange}\textbf{Ken: }\sf#1}}
\newcommand{\michal}[1]{{\color{gray}\textbf{Michal: }\sf#1}}

\def\P{\mathcal{P}} \def\C{\mathcal{C}} \def\H{\mathcal{H}}
\def\F{\mathcal{F}} \def\U{\mathcal{U}} \def\L{\mathcal{L}}
\def\O{\mathcal{O}} \def\I{\mathcal{I}} \def\S{\mathcal{S}}
\def\G{\mathcal{G}} \def\Q{\mathcal{Q}} \def\I{\mathcal{I}}
\def\T{\mathcal{T}} \def\L{\mathcal{L}} \def\N{\mathcal{N}}
\def\V{\mathcal{V}} \def\B{\mathcal{B}} \def\D{\mathcal{D}}
\def\W{\mathcal{W}} \def\R{\mathcal{R}} \def\M{\mathcal{M}}
\def\X{\mathcal{X}} \def\A{\mathcal{A}} \def\Y{\mathcal{Y}}
\def\L{\mathcal{L}}

\def\dS{\mathbb{S}} \def\dT{\mathbb{T}} \def\dC{\mathbb{C}}
\def\dG{\mathbb{G}} \def\dD{\mathbb{D}} \def\dV{\mathbb{V}}
\def\dH{\mathbb{H}} \def\dN{\mathbb{N}} \def\dE{\mathbb{E}}
\def\dR{\mathbb{R}} \def\dM{\mathbb{M}} \def\dm{\mathbb{m}}
\def\dB{\mathbb{B}} \def\dI{\mathbb{I}} \def\dM{\mathbb{M}}
\def\dZ{\mathbb{Z}}

\renewcommand{\trg}[1] {\ensuremath{g(#1)}}

\maketitle

\begin{abstract}
Multi-Agent Path Finding (MAPF) is a fundamental motion coordination problem arising in multi-agent systems with a wide range of applications.
The problem's intractability has led to extensive research on improving the scalability of solvers for it.
Since optimal solvers can struggle to scale, a major challenge that arises is understanding what makes MAPF hard. %
We tackle this challenge through a fine-grained complexity analysis of time-optimal MAPF on 2D grids, thereby closing two gaps and identifying a new tractability frontier.
First, we show that $2$-colored MAPF, i.e., where the agents are divided into two teams, each with its own set of targets, remains NP-hard.
Second, for the flowtime objective (also called sum-of-costs), we show that it remains NP-hard to find a solution in which agents have an individually optimal cost, which we call an \emph{individually optimal solution}.
The previously tightest results for these MAPF variants are for (non-grid) planar graphs.
We use a single hardness construction that replaces, strengthens, and unifies previous proofs.
We believe that it is also simpler than previous proofs for the planar case as it employs minimal gadgets that enable its full visualization in one figure.
Finally, for the flowtime objective, we establish a tractability frontier based on the number of directions agents can move in.
Namely, we complement our hardness result, which holds for three directions, with an efficient algorithm for finding an individually optimal solution if only two directions are allowed.
This result sheds new light on the structure of optimal solutions, which may help guide algorithm design for the general problem.
\end{abstract}

\section{Introduction}

Multi-Agent Path Finding (MAPF) is the problem of planning the collision-free motion of agents that operate on a graph (see formal problem definition below).
MAPF has a wide range of applications in transportation, logistics, and beyond, from warehouse automation~\cite{app:warehouses} through train scheduling~\cite{DBLP:conf/socs/AtzmonDR19}, to pipe routing~\cite{DBLP:conf/socs/BelovDBHKW20}.
The increasing proliferation of multi-agent systems has led to MAPF being a highly active research area in AI and robotics, where MAPF serves as a discrete abstraction of multi-robot motion planning.  

In this paper we study the following time optimization objectives for MAPF:
\begin{itemize}
    \item \emph{Makespan}–where we wish to minimize the time at which
the last agent reaches its goal.

    \item \emph{Flowtime} (also known as \emph{sum-of-costs})– where we wish to minimize the total time it takes the agents to reach their goals, i.e., the sum of all the agents' individual arrival times.
\end{itemize}
Arguably, these objectives are the most prevalent for MAPF.
Therefore, when using the term MAPF (or \emph{time-optimal} MAPF) in the sequel, we refer to MAPF with either of these objectives, unless stated otherwise.

Since finding optimal solutions for MAPF is intractable (see review below), the problem has been attacked along many fronts, leading to dozens of diverse algorithms (which we also call \emph{solvers}) being proposed.
Optimal MAPF algorithms include the popular Conflict-Based Search (CBS)~\cite{DBLP:journals/ai/SharonSFS15}, which has seen multiple improvements over the years (see, e.g., ~\cite{DBLP:conf/aips/FelnerLB00KK18, DBLP:conf/aaai/LiHS0K19}).
Other algorithms include Branch-and-Cut-and-Price (BCP)~\cite{BCP}, which is based on mixed-integer programming, Lazy CBS~\cite{LazyCBS}, which uses lazy-constraint programming, and the satisfiability modulo theory-based SMT-CBS~\cite{SMT-CBS}.

Over time, the development of improved algorithms has allowed effectively solving a wide range of increasingly larger instances.
Nevertheless, despite significant progress, optimal solvers can still struggle to scale.
At the same time, the performance of algorithms is not well-understood:
It has been observed that instances that do not exhibit notable differences can result in very different running times~\cite{DBLP:conf/atal/SalzmanS20}.
Furthermore, it has been challenging to empirically establish an algorithm dominating all others~\cite{DBLP:conf/socs/FelnerSSBGSSWS17}, leading to research on algorithm selection~\cite{DBLP:conf/aips/KaduriBS20} and to the search for parameters governing algorithms' performance~\cite{DBLP:conf/atal/EwingRKSA22}.

Alongside the incomplete understanding of MAPF algorithms' performance, a more fundamental question has been raised on what makes the problem itself hard~\cite{DBLP:conf/socs/FelnerSSBGSSWS17, gordon2021revisiting}. 
Gaining a deeper understanding of the problem's computational complexity has been recently highlighted as a major open challenge~\cite{DBLP:conf/atal/SalzmanS20}.
A major driving force in this work is to deepen such understanding, with the ultimate goal of applying it to obtain more scalable algorithms.
To this end, we provide a fine-grained complexity analysis of MAPF, %
as we outline next.

\begin{table*}[t]
\begin{tabular}{lccccccc}
\midrule
   Paper     & 2D grids   & \multicolumn{1}{c}{\begin{tabular}[c]{@{}c@{}}$\Delta =   0$ \\      (flowtime only)\end{tabular}} &\multicolumn{1}{c}{\begin{tabular}[c]{@{}c@{}} $2$-colored \\      MAPF\end{tabular}}& 3 directions & \# of agents     & Remark           &  \\
\midrule
\cite{YuPlanar}      &            & \checkmark                 & \checkmark     & n/a          & $O(m^2)$          &                  &  \\
\cite{DBLP:journals/siamcomp/DemaineFKMS19} & \checkmark & n/a                        & 3-colored      &              & $O(nm)$           & Makespan only    &  \\
 (Banfi et al. 2017) %
& \checkmark &          &                &              & $O(nm^2 + n^2m)$  &  &  \\
\midrule
\textbf{this work}    & \textbf{\checkmark} & \textbf{\checkmark}  & \textbf{\checkmark  }   & \textbf{\checkmark }  & \textbf{$m$ }              &                  & 
\end{tabular}
\caption{Comparison of previous hardness results for planar graphs. %
 The "\# of agents" column shows the number of agents used in the MAPF instance output by each reduction in terms of the size of the input 3-SAT formula, where $n$ and $m$ are the number of variables and clauses, respectively.
 For completeness, we note here that proofs for the non-planar case (see references in main text) generally also show hardness for $2$-colored MAPF and for the decision question of whether $\Delta=0$. Such proofs may be simpler at the expense of not being applicable to planar environments.
 Remark: The number of agents and directions used in each construction are worst case based on the details given in each proof. They may be lower following a more careful analysis of the respective proofs.}
\label{tab}
\end{table*}

\subsection{Previous work and remaining questions} %
The complexity of time-optimal MAPF has been extensively studied. Earlier hardness results apply mostly to non-planar graphs (i.e., graphs that cannot be drawn in the plane such that edges do not cross)~\cite{DBLP:conf/aaai/Surynek10, YuGraphs, DBLP:conf/aaai/MaTSKK16}.
A significant downside of such proofs is that they do not hold (i.e., cannot be easily adapted) for the more typical use cases of MAPF, which concerns planar environments.
Benchmarks and empirical performance evaluations typically revolve around the even more special case of 2D grids, which serve as an abstraction for structured environments such as warehouses.
Consequently, a major theme of subsequent complexity research has been to obtain results for more grid-like environments.
Complexity analysis for restricted cases is valuable as it can indicate whether the special structure of such cases can be exploited to obtain improved algorithms.
We therefore echo this theme and focus on the planar case.

Indeed, later hardness results have been presented for planar graphs~\cite{YuPlanar}, finally culminating with hardness for 2D grids in the works of~\cite{DBLP:journals/ral/BanfiBA17} and \cite{DBLP:journals/siamcomp/DemaineFKMS19} (the latter applies only to makespan minimization).
As for recent results, a highly restricted variant where agents must simply follow given straight paths on a 2D grid, without an optimization objective, has been shown to be NP-hard~\cite{fixed-path}.
An even newer result (which we became of aware of during the review process) shows that makespan minimization on 2D grids is NP-hard even for a fixed makespan value and fixed-parameter tractable when parameterized by the number of agents, assuming no obstacles~\cite{MAPF-FPT}.

Nevertheless, some open questions remain within the prevalent MAPF formulation for 2D grids, which we now highlight.

\paragraph{Colored MAPF.}
In practice, it may be more natural to consider the agents as being divided into teams rather than viewing each agent as unique.
In the $k$-colored MAPF variant~\cite{coloredMAPF, DBLP:journals/ijrr/SoloveyH14}, there are $k$ teams of agents, each with its own set of targets.
A valid solution then involves having each target be reached by any of the agents in the team.
$k$-colored MAPF generalizes standard (labeled) MAPF, where each agent may be viewed as a singleton team.
On the other end of the spectrum, we have \emph{unlabeled} or \emph{anonymous} MAPF, where all the agents are viewed as being part of one team and a target may be assigned to any agent.
The unlabeled case can be efficiently solved by a reduction to network flow~\cite{YuUnlabeledMAPF}.
Therefore a natural question is whether efficient solutions exist for other small $k$ values.
Previously, \cite{YuPlanar} has shown that on planar graphs the problem is NP-hard already for $k=2$.
On 2D grids, only the NP-hardness for $k=3$ for makespan has been shown~\cite{DBLP:journals/siamcomp/DemaineFKMS19}.

Therefore, our first question \textbf{(Q1)} is whether $2$-colored MAPF on 2D grids is NP-hard.
We point out that hardness has been conjectured but ultimately left unresolved~\cite{DBLP:journals/ral/BanfiBA17}.

\begin{figure}[t]
    \centering    \includegraphics[width=0.65\linewidth]{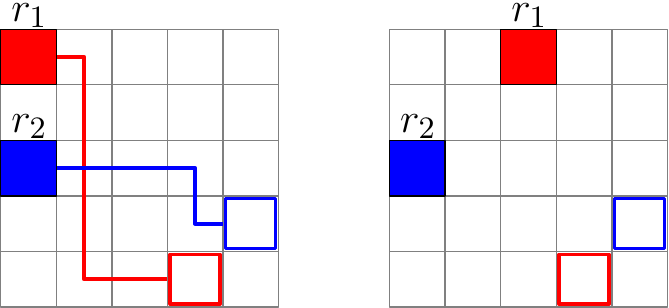}
    \caption{Two MAPF instances, where start cells are filled and target cells are unfilled.
    Left: An instance having an individually optimal solution, for which the paths are shown.
    Right: An instance that does not have such a solution, as any combination of individually optimal paths for $r_1$ and $r_2$ would result in a collision.}
    \label{fig:example_instances}
\end{figure}

\paragraph{Finding agent-wise individually optimal solutions.}
For the flowtime objective, we also investigate the complexity of finding a solution whose cost equals the lower bound cost.
A \emph{lower bound cost} for a MAPF instance is defined as the sum of costs of all agents' individually optimal paths, which are optimal paths when all other agents are ignored.
We call such a solution an \emph{individually optimal solution}.
In such a solution every agent moves at every time step along some shortest path to its goal until reaching it.
See \Cref{fig:example_instances} for examples.

Alternatively stated, if we let $\Delta$ denote the difference between the cost of the optimal solution and the lower bound cost of a given MAPF instance, the problem is to decide whether $\Delta=0$.
Note that this problem is easier than optimally solving MAPF, since finding an optimal solution also solves the former problem.

Our motivation stems from the inherent connection between $\Delta$ and the amount of computation needed to optimally solve MAPF.
In terms of search-based solvers, such as CBS, the size of the high-level search tree may grow exponentially with $\Delta$. %
This fact stems from the typical flow of optimal MAPF solvers, which begins by determining the existence of an individually optimal solution and only proceeds to look for more expensive solutions if no such solution exists.
Therefore, it seems natural to ask whether the initial step of deciding whether $\Delta=0$ can be done more efficiently.

The answer for planar graphs has been negative since instances in the hardness proof of~\cite{YuPlanar} are yes-instances if and only if $\Delta=0$. %
However, we cannot say the same for the proof for 2D grids~\cite{DBLP:journals/ral/BanfiBA17} since their construction does not bound $\Delta$ in yes-instances.\footnote{In yes-instances in their proof the clause agents may have to take detours, resulting in
paths longer than their individually optimal paths. Hence the answer to the decision question in their proof does not indicate whether $\Delta=0$.}

We therefore pose the following question \textbf{(Q2)}:
For MAPF with the flowtime objective on 2D grids, is it NP-hard to determine the existence of an individually optimal solution? %

Note we can define an equivalent question for the makespan objective, though its answer is immediate since for makespan any existing construction can be easily modified to show hardness for this case.\footnote{Suppose that there is a hardness construction showing that deciding whether a solution with makespan better than $\M$ exists. Then we may modify the construction by adding a single agent whose target is $\M$ units away in a separate part of the construction from the rest of the agents. The lower bound for makespan becomes $\M$ due to the new agent.}

\subsection{Contribution}
We provide a new NP-hardness proof for MAPF on 2D grids that settles the hardness of the restricted variants in Q1 and Q2. %
Namely, $2$-colored MAPF and finding an (agent-wise) individually optimal solution remain NP-hard on 2D grids.
Our results are tight in the following sense:
For Q1 the case of $k=1$ (i.e., anonymous MAPF) is efficiently solvable.
Similarly, Q2 considers the smallest value for $\Delta$.
Thus, our proof replaces and unifies multiple previous proofs, including those of~\cite{YuPlanar, DBLP:journals/ral/BanfiBA17, DBLP:journals/siamcomp/DemaineFKMS19}, while obtaining stronger results that close two gaps in MAPF's complexity.

Moreover, our proof is simpler in terms of using fewer agents than previous proofs for the planar case;
see \Cref{tab} for a comparison with previous proofs.
The relatively low number of agents in our reduction stems from the fact that we only use one type of agent, as opposed to multiple types used in previous reductions.\footnote{By "types of agents" we refer to agents having a distinguished role in the proof, e.g., agents representing SAT clauses versus agents representing variables.}
Consequently, this makes it (arguably) easier to visualize our complete construction, with all its constituent components, on the 2D grid.%

Finally, we establish a new tractability frontier for MAPF based on the number of cardinal directions agents can move along.
Our proof uses only three such directions, as opposed to previous 2D grid proofs, which use all four directions.
We complement the hardness results by showing that if we allow only two directions of motion, we can efficiently find an individually optimal solution for flowtime (if such a 
solution exists).
While this positive result is quite restricted, it sheds new light on the structure of optimal solutions.
Namely, we prove that yes-instances have a canonical solution that can be found by a simple prioritized planning-based algorithm.
Overall, by bridging complexity analysis and algorithm design we hope to inspire enhanced reasoning techniques that can further improve MAPF solvers.

\section{Problem Definition}\label{sec:definitions}
In Multi-Agent Path Finding (MAPF)~\cite{DBLP:conf/socs/SternSFK0WLA0KB19} we are given a graph~${G=(V,E)}$ and a set of agents~${\left\{r_1, \ldots, r_N\right\}}$, where each agent $r_i$ has a start and a goal vertex, denoted by $(s_i,g_i)$ respectively, where ~${s_i,g_i \in V}$.
Time is discretized and at each time step an agent can either \emph{wait} at its current vertex or \emph{move} across an edge to an adjacent vertex.
A \emph{timed path} (or simply \emph{path)} is a sequence $(v_1,\ldots, v_t)$ of vertices representing the current location of an agent at each time step, i.e., for each $i<t$, $v_i$ and $v_{i+1}$ are either adjacent or the same vertex.
A \emph{feasible} or \emph{collision-free} solution is a set of paths~${\mathcal{P}=\left\{p_1,p_2,...,p_N\right\}}$ such that $p_i$ is a timed path for agent $r_i$ from $s_i$ to $g_i$, and there are no conflicts between any two paths in $\mathcal{P}$.
Namely, the following types of conflicts do not occur:
A \emph{vertex-conflict}, in which two agents occupy the same vertex at the same time step.
An \emph{edge-conflict}, in which two agents traverse the same edge in opposite directions at the same time step.

An \emph{optimal} solution is a set of paths~$\mathcal{P}$ that also optimizes some objective function.
For time-optimal MAPF formulations, the most common objectives are the following:
\textit{Makespan}: the maximum number of time steps needed for all agents to reach their targets, i.e., $\max_{i=1,\ldots,N} |p_i|$.
\textit{Flowtime} (also known as \emph{sum-of-costs}): the sum of time steps needed by all agents to reach their targets, i.e., $\sum_{i=1,\ldots,N} |p_i|$.

\paragraph{k-colored MAPF.} In $k$-colored MAPF the agents are divided into $k$ teams and each team is associated with a number of targets that is identical to the number of agents in the team.
A feasible solution brings the agents in each team to the team's respective targets, i.e., each target can be occupied by any agent in the team. Note that $N$-colored MAPF, i.e., each team consists of a single agent, is the same as the original definition of MAPF above. The $1$-colored variant is called \textit{unlabeled} or \textit{anonymous} MAPF.
Colored MAPF has also been studied under the name Target Assignment and Path Finding (TAPF)~\cite{TAPF}.

\paragraph{Other conflict types.}
One may define other types of conflicts that needed to be avoided in a feasible solution:
A \emph{following conflict} occurs when an agent attempts to move to a vertex that is being vacated by another agent at the same time step.
A \emph{cycle conflict} occurs when multiple agents $r_1,\ldots,r_k$ are involved in a following conflict such that $r_{i (\mod k)}$ attempts to move to the current vertex of $r_{i+1 (\mod k)}$, i.e., a rotating cycle of agents is formed.

\section{Hardness Results}
In this section we present our hardness results, which are as follows:

\begin{theorem}
For the makespan and flowtime objectives, MAPF remains NP-complete on 2D grids with holes even when the agents are restricted to move along $3$ directions and for the following variants:
    \begin{enumerate}[label=(\roman*)]
    \item
    $2$-colored MAPF, i.e., when there are two groups of agents, where agents within each group are considered interchangeable.
    \item
    Deciding whether there is an individually optimal solution for flowtime.
    \end{enumerate}
    \label{thm:hardness}
\end{theorem}

\begin{figure}[t]
    \centering    \includegraphics[width=0.64\linewidth]{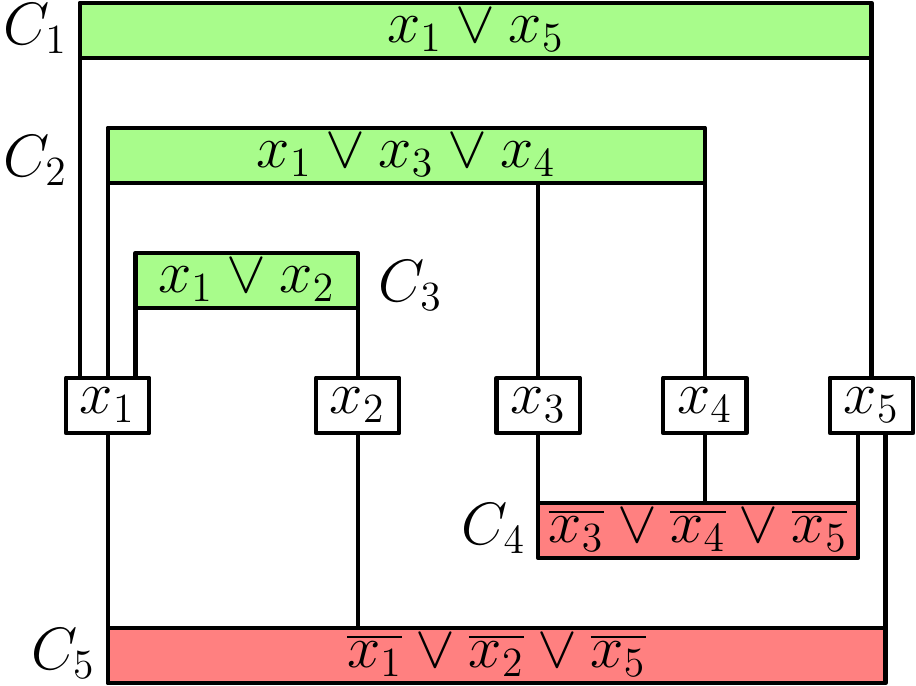}
    \caption{\textsc{A Monotone Planar 3-SAT} instance in a rectilinear embedding. } %
    \label{fig:mpsat_example}
\end{figure}

Since membership in NP is straightforward~\cite{YuGraphs}, we only show hardness.
The base problem for our reduction is \textsc{Monotone Planar 3-SAT}, which we now define.

\paragraph{\textsc{Monotone Planar 3-SAT}.}
Let $\phi = \bigwedge C_i$ be a \textsc{3-SAT} formula having $n$ variables and $m$ clauses, where each clause $C_i$ is the disjunction of at most three literals, each of which is either a variable or its negation.
We now define the two special elements of this version of 3-SAT.

A clause is \emph{positive (resp. negative)} if it contains only positive (resp. negative) literals.
An instance of \textsc{3-SAT} is \emph{monotone} if it only has positive or negative clauses.

For the planarity element, we consider the bipartite graph $G_\phi$ that contains a vertex for each variable and for each clause, and has an edge between a variable vertex and a clause vertex if and only if the variable appears in the clause.
In \textsc{Planar 3-SAT}, which remains NP-complete~\cite{DBLP:journals/siamcomp/Lichtenstein82}, we require $G_\phi$ to be planar.

The graph $G_\phi$ of a \textsc{Planar 3-SAT} instance can be drawn in a \textit{rectilinear} embedding as follows~\cite{DBLP:journals/siamdm/KnuthR92}:
All vertices are drawn as unit-height rectangles, with all the variable-vertex rectangles centered on a fixed horizontal strip called the \textit{variable row}.
Every edge is a vertical line segment, which we call a \textit{leg}, that does not cross any rectangles, as shown in Figure~\ref{fig:mpsat_example}. %
\textsc{Planar 3-SAT} remains NP-complete when $G_\phi$ is given as a rectilinear embedding.

The 3-SAT version we use combines planarity with monotonicity as follows:
\textsc{Monotone Planar 3-SAT} is a restriction of \textsc{Planar 3-SAT} to monotone instances in which the rectilinear embedding of $G_\phi$ has all positive clause-vertices above the variable row and all negative clause-vertices below it, as illustrated in Figure~\ref{fig:mpsat_example}.
\textsc{Monotone Planar 3-SAT} remains NP-complete~\cite{DBLP:conf/cocoon/BergK10}.

\paragraph{Nesting level and root clauses.}
We will use the following hierarchical relationship between clauses, which is induced by a fixed rectilinear embedding of $G_\phi$:
Let $C$ and $C'$ be two clauses that are on the same side of the variable row.
We say that $C$ \emph{encloses} $C'$ if one can draw a vertical line segment $s$ connecting $C$ and $C'$, and also $C$ is vertically further from the variable row than $C'$.
$C$ is called the \emph{parent} of $C'$ if it encloses $C'$ and furthermore we can draw $s$ without crossing any clauses, i.e., $C$ directly encloses $C'$.
In this case, we call $C'$ a \emph{child} of $C$.
A clause that does not have a parent is called a \emph{root} clause.
For example, in Figure~\ref{fig:mpsat_example} $C_1$ encloses $C_3$, but is only the parent of $C_2$.

\begin{figure}[t]
    \centering    \includegraphics[width=0.66\linewidth]{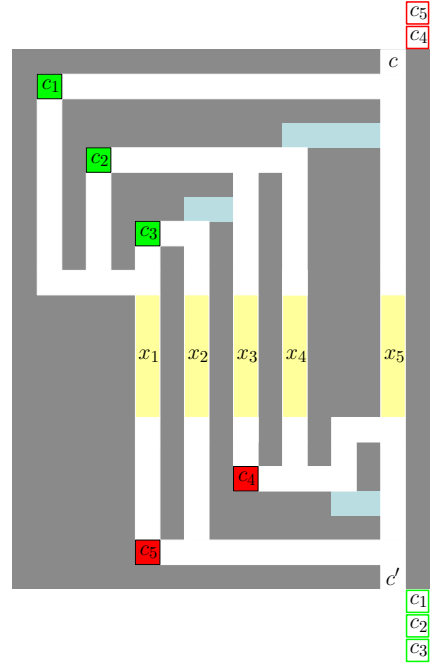}
    \caption{The instance $M$ in our reduction, which corresponds to the 3-SAT formula shown in \Cref{fig:mpsat_example} (agent $c_i$ corresponds to the clause $C_i$).
    Positive (resp. negative) agents and their targets are green (resp. red).
    The two colors represent the two teams in the corresponding $2$-colored MAPF instance.
    Obstacle cells are dark gray.
    The lengths of vertical corridors in the figure are slightly smaller than required -- see text for details.}
    \label{fig:construction}
\end{figure} 

We may assume without any loss of generality that there is a single root clause on each side of the variable row~\cite{DBLP:conf/soda/AgarwalAGH21}, as occurs in Figure~\ref{fig:mpsat_example}, where $C_1$ and $C_2$ are the two root clauses.

We now recursively define the nesting level of a clause using the hierarchical structure of $G_\phi$:
Any non-parent clause has a nesting level of $0$.
Any parent clause $C$ has a nesting level of $\ell + 1$, where $\ell$ is the maximum level of any child of $C$.
For example, in Figure~\ref{fig:mpsat_example}, $C_1$ has a nesting level of $2$.

\paragraph{Reduction Outline.}
Given a Monotone Planar 3-SAT instance $\phi$ we construct a MAPF instance $M \coloneqq M(\phi)$ that has an individually optimal solution, for the flowtime objective, if and only if $\phi$ is satisfiable.
The resulting instance will allow proving the hardness of all the considered MAPF variants, up to minor changes that we describe later.
We use the rectilinear embedding of $G_{\phi}$, which is easily embedded in the 2D grid.
Roughly speaking, we turn vertices and edges in this embedding into horizontal and vertical unit-width corridors, respectively.
Such corridors are formed by filling up most of the grid with obstacle cells. Refer to \Cref{fig:construction}, which will be described in more detail later.
We represent each clause in $\phi$ by a corresponding agent as follows, resulting in $m$ agents in total:
For a positive clause $C$ in $\phi$, we have a respective \emph{positive agent} that starts in $C$'s embedding and has to reach its target by moving down and right, crossing the variable row.
Symmetrically, for a negative clause $C$ in $\phi$, we have a respective \emph{negative agent} (also starting in $C$'s embedding) that has to reach its target by moving up and right.

The key idea in the construction is to force the unidirectional travel of agents through \emph{variable channels}, which are long vertical corridors that represent variables (yellow in \Cref{fig:construction}).
Our paths are constructed so that each agent is forced to pass through one of the variable channels corresponding to variables appearing in the clause.
This creates a point of contention along variable channels, for which we only have uni-directional motion in an individually optimal solution. The direction of motion in such a solution corresponds to the variable's assignment in $\phi$.%

The overall idea is similar to the proof of~\cite{YuPlanar}, where MAPF was first shown to be NP-complete for planar graphs. We therefore also use some of the same terminology.

\paragraph{The details.}
We now discuss $M$ in more detail.
We convert the rectilinear embedding of $\phi$ to $M$ as follows.
Each clause is converted to a clause gadget with the following path segments:
The horizontal rectangle (clause vertex in $G_{\phi}$) becomes a horizontal corridor of unit height and each leg (edge) becomes a vertical corridor of unit width.
Each positive (resp. negative) agent is initially located in the top (resp. bottom) left cell of its corresponding gadget.

Each variable gadget is a long vertical path (one unit wide) of length $L$ (to be specified later), which forms the variable channel.
All the variable channels start and end at the same grid row.
Clause gadgets are connected to variable channels as follows.
Each variable has two horizontal corridors, one above the variable row and one below it. %
On each side of the variable row, all the legs that lead to the variable go the horizontal corridor of that variable (on the respective side).
The two horizontal corridors connect to the variable channel from the left.
For example, in \Cref{fig:construction}, for $x_1$ we have three incoming legs to the respective horizontal corridor on the positive side, making it $5$ units wide.

Note that the rightmost leg entering a variable (on each side of the variable row) is located in the same column as that variable's channel. If a certain literal only appears once, then the horizontal corridor is only one unit wide (such as for all variables except $x_1$ on the positive side of the variable row in \Cref{fig:construction}).

Before specifying the agents' target cells, we describe the \textit{backwards paths}.
We add such paths to $M$ to facilitate the motion of agents on the opposite vertical side of construction with respect to their start cells.
Let $C_i$ and $C_j$ be two positive clauses such that $C_i$ is the parent of $C_j$.
The backwards path for this pair of clauses starts at the top right corner of $C_i$'s gadget and continues rightward until reaching a leg of $C_j$'s gadget.
This placement ensures that only negative agents can use the path (as a positive agent would have to detour from its individually optimal path to use the backward path).
Backwards paths below the variable row are handled symmetrically.
These paths are shown in pale cyan in \Cref{fig:construction}.

To enable negative agents to reach their targets, we create an \emph{opening} from the positive root clause gadget $\G$ that leads to their targets.
The opening consists of one cell neighboring the top right cell of $\G$, which we denote by $c$ (the cell is marked so \Cref{fig:construction}).
All targets of negative agents are placed above and to the right of this opening, e.g., by placing targets one above another.
The targets of positive agents are handled symmetrically.
We denote the corresponding opening cell on the bottom side of the construction by $c'$.

Let $W$ be the minimum number of columns required for the construction. %
We have $W \le 6m$ by assuming $3$ legs per clause and an empty cell between every two adjacent legs.
We define the \textit{height} of a clause gadget as the length of its legs.
We adjust the height of each clause as follows: %
We set the height of clauses with a nesting level of $i$ to be $iW$.
Setting the heights is easily possible by adjusting the lengths of the legs.
With this placement, it is straightforward to verify that each positive (resp. negative) agent's start cell has a unique (Manhattan) distance to $c'$ (resp. $c$), which is a key property we later use.

We now set $L$, the length of the variable channels, to be the maximum distance required by any agent to enter one of its respective variable channels, over all the agent's individually optimal paths.
For example, for $c_3$ the maximum distance is 5 (going right and then down).
Observe that $L$ is $O(m^2)$ (and is roughly equal to the maximum height of any clause gadget). %

\paragraph{Size of the construction.}
We give a rough upper bound on the size of the construction.
Each clause is composed of at most 4 horizontal/vertical corridors, each of length at most $O(m^2)$ and has backwards path of length at most $O(n)$.
There are $n$ variable channels of length $O(m^2)$.
Therefore $O(m^3 + nm^2)$ cells suffice to form $M$ and thus the reduction can be performed in polynomial time.
We conjecture that a more compact construction is possible.

We now turn to the correctness of the reduction and thus prove \Cref{thm:hardness}.
\begin{proof}[Proof of \Cref{thm:hardness}]
We first prove statement (ii) by showing $\phi$ has a satisfying assignment if and only if $M$ has an individually optimal solution, for flowtime.
In particular, this statement also establishes the hardness of optimizing flowtime.

Given a satisfying assignment to $\phi$ we define an individually optimal solution to $M$.
We describe the motion of a positive agent as negative agents move in a symmetric manner. %
Each positive agent moves using down and right moves towards its target at every time step and only stops when reaching its target.
Each such agent is able to travel along a variable channel corresponding to a variable assigned true as per the satisfying assignment.
Upon reaching the other end of a variable channel, the agent continues down, entering a leg of a negative clause gadget, say $C_i$, and then right until reaching a backwards path.
The backwards path takes the agent to the leg of the gadget of $C_i$'s parent clause.
The agent can continue traveling similarly through the parent gadget, using additional backwards paths, until reaching the gadget of the negative root clause.
From there, it easily reaches its target via the opening cell.
All the possible variable channels for an agent result in the same path length for that agent. %

We now prove that the described solution is collision-free.
First, note that the length of the variable channels is sufficiently long so that negative agents cannot collide with positive agents. This is true, since after $L$ time steps every agent is inside some variable channel and each channel contains only positive agents or only negative agents.

Therefore, we turn to collisions among positive agents (the same arguments apply to negative agents).
Recall that $c'$ denotes the opening cell of the negative root clause gadget.
By construction, each start cell of a positive agent has a unique (Manhattan) distance to $c'$. %
At each time step, all the positive agents that 
have not reached $c'$ get close to $c'$ by one unit, since they all move down or right.
Hence, such agents maintain the property of having a unique distance from $c'$.
Therefore, they cannot collide, as a collision requires two agents to be located the same distance from $c'$.
After an agent visits $c'$, it continues to its target and does not block other agents from reaching their target later on, as per the solution.

In the other direction, let us assume that $M$ has an individually optimal solution.
Since agents must move at every time step (along a shortest path) until reaching their respective targets, the solution cannot have both a positive agent and a negative agent traverse the variable channel.
This is true because an agent can reach a variable channel in at most $L$ time steps.
Since the length of a variable channel is $L$, if it is used by both a positive and negative agent, one of them would have to wait for the other.
Therefore we simply set a variable in $\phi$ to be true if and only the corresponding variable channel is used by a positive agent in the given solution.
Such an assignment satisfies $\phi$ since each clause's agent must use a variable channel corresponding to one of the variables in the clause, which sets the variable to the appropriate value.

Next, we prove hardness for the makespan objective.
We do so by first slightly changing $M$.
We ensure that all agents have the same individually optimal path length as follows:
Let $d$ denote the maximum individually optimal path length over all agents in $M$.
We now make every target be at a distance $D$ from its respective start cell by moving each target rightward a sufficient number of cells so that the condition holds.
Let $M'$ denote the resulting MAPF instance.
It is straightforward to verify that $M'$ has a solution with a makespan of $d$ if and only if $M$ has an individually optimal solution for flowtime.

Finally, to prove statement (i) for both objectives, we use the same constructions $M$ and $M'$, with the positive agents as one group and the negative agents as the other group.
It is straightforward to verify that the statement follows.
\end{proof}

We remark that the proof holds for all common conflict types in MAPF (see problem definition above) since it relies on head-on collisions only.

\begin{figure}[t]
    \centering    \includegraphics[width=0.34\linewidth]{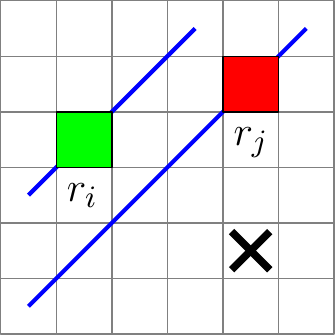}
    \caption{Illustration of the key observation (see preliminaries of our algorithm) with the two diagonals marked. Agent $r_j$ would always be closer to a potential point of conflict (marked by a cross) than $r_i$.}
    \label{fig:diff_diags}
\end{figure}

\section{Efficiently Finding Individually Optimal Solutions for a Restricted Case}
In this section, we establish a tractability frontier for finding an individually optimal solution for the flowtime objective.
Our hardness proof relies on three directions of motion to create contention in narrow corridors between agents moving in opposite directions.
We now show that if we further restrict the problem to only two directions of motion, we can efficiently solve this MAPF variant:

\begin{theorem} \label{thm:alg}
Let $G$ be a subgraph of the 2D grid, i.e., a 2D grid with holes.
For MAPF with the flowtime objective on $G$ where agents are restricted to move along two directions, there is an algorithm that finds an individually optimal solution or reports that none exists.
The algorithm runs in time $O(N|V|)$, where $N$ is the number of agents and $|V|$ is the number of cells in $G$.
\end{theorem}

We now describe the algorithm and prove its correctness.

\paragraph{Preliminaries.}
Let us assume without any loss of generality that the allowed directions of motion are down and right (the problem is trivial for two opposite directions).
Furthermore, we assume that each agent's target is located either below or to the right of its start cell (or both), since otherwise the algorithm should clearly report "no solution".

Let us partition the agents into sets based on the grid diagonal on which they start.
That is, two agents with starting cells having coordinates $(x_i,y_i)$ and $(x_j,y_j)$ are in the same \emph{diagonal set} if and only if $x_i-x_j = y_i-y_j$.
Let $D_1, \ldots, D_r$ denote the resulting diagonal sets in left to right order (according to the $x$-intercept of the corresponding grid diagonal).

We use the following key observation:
Let $r_i, r_j$ be two agents such that $r_i \in D_a, r_j \in D_b$, $a < b$.
In an individually optimal solution, $r_i$ and $r_j$ can collide only at \trg{j}, since for any other cell $c$ where their paths intersect $r_j$ will leave $c$ the time step before $r_i$ enters $c$.
In other words, $r_i$'s path must simply avoid $r_j$'s target and need not take into account $r_j$'s path at all.
See \Cref{fig:diff_diags} for an illustration.
Therefore, our algorithm can essentially handle each diagonal set independently and treat the targets of agents from diagonal sets that are to the right as obstacles.

\begin{figure}[t]
    \centering    \includegraphics[width=0.74\linewidth]{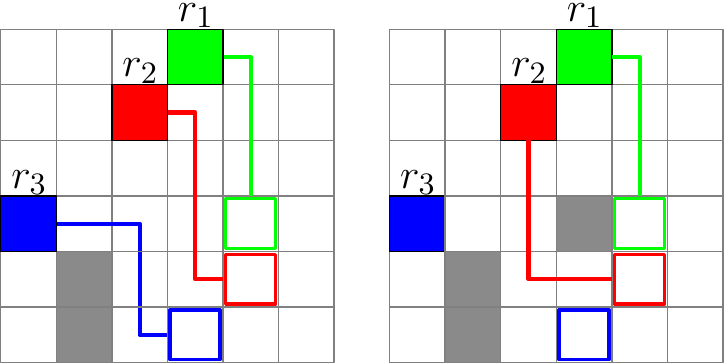}
    \caption{Left: An instance consisting of agents on a single diagonal where the algorithm finds a solution. Right: When one obstacle cell is added, the left instance becomes a "No" instance.
    In both cases, our algorithm plans paths in the order $r_1$, $r_2$, $r_3$. Observe that preferring right over down is crucial for finding a solution for the left instance. For the right instance, paths are found for $r_1$ and $r_2$, but the algorithm concludes that it is not possible to also find an individually optimal path for $r_3$.}
    \label{fig:run_ex}
\end{figure}

\paragraph{Algorithm.}
Our algorithm can be viewed as a combination of two simple elements: prioritized planning using a special ordering of the agents and a tie-breaking rule for the single-agent path search.
We first give a succinct description of the algorithm along such terms and then provide an explicit description.

Recall that in prioritized planning~\cite{DBLP:conf/aiide/Silver05} we plan the agents' paths one by one, in a decoupled manner, where each agent's path must not conflict with previously planned paths.
We apply prioritized planning by planning paths for agents in this manner, in right to left order of their start positions (with ties broken arbitrarily).
Furthermore, we require that each single agent path search prefers going right before going down (e.g., a path planned without any obstacles would go right and then down).
During this process, if the optimal path found for a single agent is not individually optimal, we report a "No" instance.
Otherwise, we clearly get an individually optimal solution.
We shall see that both the ordering of the agents and the tie-breaking in each agent's path search are required for correctness.

We give an alternative, more explicit description of the algorithm that can provide better intuition for its correctness.
The explicit version of the algorithm plans the paths for each diagonal set $D_i$ separately, in the order $i=r,r-1,\ldots,1$.
Therefore, we fix $D_i$ and describe the planning for it.

First, for all $j > i$ we \emph{remove} from $G$, i.e., mark as obstacles, the target cells of the agents in $D_j$.
The removed cells cannot be used by agents in $D_i$ in an individually optimal solution since such cells would be occupied by agents from diagonals right of $D_i$.

Let $r_1,\ldots,r_{u}$ denote the agents sorted along $D_i$ according to decreasing $x$-value of their starting cells.
We then go over each $r_j$ in this order and attempt to find an individually optimal path for $r_j$ as follows.
We run a DFS from the start cell of $r_j$ that prefers visiting cells to the right of the currently visited cell.
If no such path is found within the bounding box of the start and target cell of $r_j$, report a "No" instance.
Otherwise, let $p_j$ denote the resulting path.
We remove the cells appearing on $p_j$ from $G$ and proceed to the next agent.
The removed cells cannot be used by other agents in $D_i$ since that would result in a collision.

Once finished with $D_i$, add back all the removed cells of $p_1, \ldots, p_{u}$ to $G$ and proceed to the next diagonal.
An example of the algorithm handling a single diagonal $D_i$ is shown in \Cref{fig:run_ex}.

\paragraph{Correctness.}
First, we note that if the algorithm finds a solution then it is clearly individually optimal and collision-free. Therefore, if no such solution exists, the algorithm will report a "No" instance.
Hence, it remains to prove that if an individually optimal solution exists, the algorithm will find such a solution.

In what follows, we use the following key definition.
Given a path $p$, we denote by $U(p)$ the union of cells that either appear in $p$ or are above some cell in $p$. %
A path $q$ is said to be \emph{weakly above} $p$ if the cells of $q$ are contained in $U(p)$.
See Figure~\ref{fig:weakly_above}.
We now show that the existence of an individually optimal solution implies the existence of a canonical solution identical to the one found by the algorithm.
In this canonical solution, each path $p_j$ for agent $r_j$ is weakly above the agent's path in any individually optimal solution.

\begin{figure}[t]
    \centering    \includegraphics[width=0.68\linewidth]{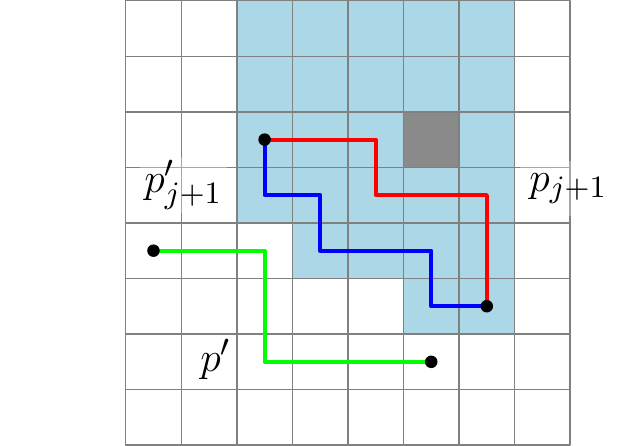}
    \caption{An illustration of the paths $p'$ (green), $p'_{j+1}$ (blue), and $p_{j+1}$ (red) from \Cref{lem:alg}.
    The cells of $U(p'_{j+1})$ are shaded and $p_{j+1}$ is weakly above $p'_{j+1}$.}
    \label{fig:weakly_above}
\end{figure}

Let us fix some diagonal set $D=r_1,\ldots,r_u$, where the $r_i$'s appear in decreasing $x$-value of their starting cells.
We now assume that an individually optimal solution exists and denote the paths for the agents of $D$ in this solution by  $P'=\set{p'_1, \ldots, p'_u}$.
We denote by $P=\set{p_1, \ldots, p_u}$ the corresponding paths in the solution found by the algorithm.
Let $P_j$ denote a solution composed of the first $j$ paths of $P$ and the remainder of the paths are from $P'$, i.e., $P_j \coloneqq \set{p_1, \ldots, p_j, p'_{j+1}, \ldots, p'_u}$. %
Roughly speaking, the following lemma shows that we may transform the paths in $P'$ one by one to the paths of $P$, the canonical solution, while maintaining a collision-free solution.

\begin{lemma} \label{lem:alg}
For every $j = 0,\ldots,u$, the algorithm succeeds in finding the paths $p_1, \ldots, p_j$ and the solution $P_j$ is collision-free.
\end{lemma}

\begin{proof}
The proof is by induction on $j$. The base case of $j=0$ is trivial since $P_0 = P'$.

Let us now assume that the lemma holds for some $j$.
That is, the algorithm has found paths for the first $j$ agents and the path $p'_{j+1}$ of $P'$ does not conflict with the paths of these agents.
The next agent of $D$ for which the algorithm plans a path is $r_{j+1}$.
Since this path search prefers right moves before downward moves, it considers all paths that are weakly above $p'_{j+1}$ before eventually considering $p'_{j+1}$ as an option.
Therefore, the search for $p_{j+1}$ will succeed, resulting in a path weakly above $p'_{j+1}$ (which may be the same as $p'_{j+1}$).

It remains to prove that $p_{j+1}$ does not conflict with paths in $\set{p'_{j+2}, \ldots, p'_u}$. Let $p'$ be a path in $\set{p'_{j+2}, \ldots, p'_u}$.
Since $p_{j+1}$ is weakly above $p'_{j+1}$, the cells along $p_{j+1}$ are contained in $U(p'_{j+1})$.
Furthermore, since $p'$ does not conflict with $p'_{j+1}$, the cells along $p'$ are disjoint from $U(p'_{j+1})$;
See \Cref{fig:weakly_above}.
Therefore, combining the last two facts, $p'$ does not conflict with $p_{j+1}$, which means that the lemma holds for $j+1$.
\end{proof}

For $j=u$ the lemma implies that the algorithm succeeds in finding paths for the agents of $D$. Finally, the lemma holds for all diagonal sets, thus establishing the correctness of the algorithm. \Cref{thm:alg} follows.

\section{Discussion and Conclusion}

We have provided a fine-grained complexity analysis of time-optimal MAPF on 2D grids, thereby closing two gaps in the literature.
For the flowtime objective, our analysis reveals a tractability frontier concerning the number of directions the agents can move in.
While our positive result holds for quite a restricted case, it sheds new light on the complexity of MAPF that may have further implications, which we now discuss.

In \cite{YuPlanar}, agents that move in opposite directions are used to show the hardness of MAPF for various distance and time objectives.
In turn, Yu concludes that the problem becomes hard in the presence of contention
that occurs when two or more groups of agents need to move in opposite directions through the same set of narrow paths.
One natural follow-up question is whether the problem is hard without such opposite-direction movement. %
A "no" answer here, i.e., in the form of an efficient algorithm for a restricted MAPF variant, may allow speeding up MAPF algorithms by incorporating its reasoning in solvers for the general problem. %

This question has been recently addressed in \cite{DBLP:conf/atal/GeftH22} for distance-optimal MAPF (where the goal is to minimize the total traveled distance, ignoring time) on 2D grids.
They show that distance-optimal MAPF is NP-hard even if the agents need to go in the same general direction, specifically when only allowed to move down and right.
Furthermore, it may be verified that their proof shows that in this case finding an individually optimal solution is already NP-hard. %
In this context, our paper somewhat surprisingly reveals that the same question for the flowtime objective is tractable.
While the complexity of the problem without restricting $\Delta$ remains open for flowtime, our algorithm suggests that MAPF with agents moving along two directions (or similar instances where most agents move in the same general direction) may be easier to solve efficiently. %

From the perspective of parameterized complexity~\cite{cygan2015parameterized}, which has been highlighted as a research avenue for tackling MAPF's hardness~\cite{DBLP:conf/atal/SalzmanS20}, our analysis provides guidance concerning $\Delta$ as a parameter. %
Here $\Delta$ could be a good candidate due to empirical observations showing that sometimes it can be much smaller than the number of agents. %
Our hardness proof rules out $\Delta$ as a candidate for 2D grids in the general case.
That is, an algorithm with the running time $f(\Delta) \cdot \mathrm{poly}(|M|)$, where $|M|$ is the size of the MAPF instance, is unlikely to exist.\footnote{This is true, since otherwise, we could run the algorithm on MAPF instances from our reduction and efficiently solve NP-complete problems.}
On the other hand, we cannot say the same for the restricted case where the agents are moving in the same general direction, such as down and right, where our restricted positive result calls for further investigation.

We point out that the class of instances in which each agent's target is either below or to the right (or both) of its start cell can be empirically challenging.
This is easily seen through the issue of rectangle symmetries, which arise in instances of this class containing only two agents~\cite{DBLP:conf/aaai/LiHS0K19}.
Before being explicitly addressed, such symmetries have significantly hindered the performance of state-of-the-art solvers such as CBS, Lazy CBS, and BCP~\cite{LiPhD22}.
Therefore, a research direction currently being explored is using the insights gained from exploring this class of instances to further enhance the performance of MAPF solvers.
This includes leveraging our algorithm towards improved symmetry reasoning and lower bounding.

\section*{Acknowledgments}
This work has been supported in part by the Israel Science Foundation (grant no.~1736/19), by NSF/US-Israel-BSF (grant no.~2019754), by the Israel Ministry of Science and Technology (grant no.~103129), by the Blavatnik Computer Science Research Fund, by the Yandex Machine Learning Initiative for Machine Learning at Tel Aviv University, by the Shlomo Shmeltzer Institute for Smart Transportation at Tel Aviv University, and by the Israeli Smart Transportation Research Center (ISTRC).

The author wishes to acknowledge Dominik Michael Krupke for a discussion that inspired the hardness results and also Dan Halperin and Nathan Libman for useful discussions and comments.

\bibliography{references.bib}

\end{document}